\documentclass[10pt,conference]{ieeeconf}
\IEEEoverridecommandlockouts
\usepackage{amsfonts,amssymb,amsmath,amsthm}
\usepackage{algorithm,algorithmic}
\usepackage{array}
\usepackage{booktabs}
\usepackage{textcomp}
\usepackage{stfloats}
\usepackage{url}
\usepackage{cite}
\usepackage{verbatim}
\usepackage{graphicx}
\usepackage{balance}
\usepackage{array}
\usepackage{xcolor}
\usepackage{nomencl}

\newtheorem{lemma}{Lemma}
\newtheorem{theorem}{Theorem}

\DeclareMathAlphabet{\mathbbb}{U}{bbold}{m}{n}
\newcolumntype{M}[1]{>{\centering\arraybackslash}m{#1}}

\newcommand{\col}{\operatorname{col}}
\newcommand{\diag}{\operatorname{diag}}
\newcommand{\blkdiag}{\operatorname{blkdiag}}
\newcommand{\argmin}{\operatorname{argmin}}
\newcommand{\proj}{\operatorname{proj}}
\newcommand{\Norcone}{\operatorname{N}}

\allowdisplaybreaks[4]

\begin{document}
\title{\bf \LARGE Optimal Bidding Strategies in Network-Constrained Demand Response: A Distributed Aggregative Game Theoretic Approach}

\author{Xiupeng Chen, Jacquelien M. A. Scherpen, and Nima Monshizadeh
\thanks{Xiupeng Chen, Jacquelien M. A. Scherpen, and Nima Monshizadeh are with Jan C. Willems Center for Systems and Control, ENTEG, University of Groningen, Groningen, 9747 AG, The Netherlands. (email: {\tt \{xiupeng.chen, j.m.a.scherpen, n.monshizadeh\}@rug.nl})}
}

\maketitle

\setlength\abovedisplayskip{1.5pt}
\setlength\belowdisplayskip{1.5pt}
\setlength\intextsep{1pt}

\begin{abstract}
Demand response has been a promising solution for accommodating renewable energy in power systems. In this study, we consider a demand response scheme within a distribution network facing an energy supply deficit. The utility company incentivizes load aggregators to adjust their pre-scheduled energy consumption and generation to match the supply. Each aggregator, which represents a group of prosumers, aims to maximize its revenue by bidding strategically in the demand response scheme. Since aggregators act in their own self-interest and their revenues and feasible bids influence one another, we model their competition as a network-constrained aggregative game. This model incorporates power flow constraints to prevent potential line congestion. Given that there are no coordinators and aggregators can only communicate with their neighbours, we introduce a fully distributed generalized Nash equilibrium seeking algorithm to determine the optimal bidding strategies for aggregators in this game. Within this algorithm, only estimates of the aggregate and certain auxiliary variables are communicated among neighbouring aggregators. We demonstrate the convergence of this algorithm by constructing an equivalent iteration using the forward-backward splitting technique.
\end{abstract}


\section{Introduction}
The rising integration of renewable energy resources presents challenges for operators in matching energy supply and demand cost-effectively \cite{sinsel2020challenges}. With the growing proliferation of distributed energy resources, the so-called "prosumers" can now participate in demand response programs by adjusting their flexible generation and consumption patterns. Aggregators typically serve as intermediaries between the utility company and a vast number of prosumers, addressing scalability issues \cite{gkatzikis2013role}. Given that each aggregator represents a significant portion of the demand, it acts as a price-maker, bidding strategically when offering flexibility services to the utility company. Moreover, the revenue and feasible bid of each aggregator are influenced by the bids of all other aggregators due to their physical connection within a distribution network. This work focuses on determining the optimal bidding strategies for aggregators to maximize their profits in a demand response program.

Game theory offers a powerful framework for examining competitive behaviours in demand response programs. For instance, three distinct game-theoretic billing methods were explored for demand-side management in residential communities \cite{hupez2022pricing}. A game-theoretic horizon decomposition approach for real-time demand-side management was introduced in \cite{mishra2022game}. In \cite{scarabaggio2022noncooperative}, a non-cooperative control mechanism was devised for the optimal operation of distribution networks. Meanwhile, \cite{chen2019energy} featured a novel demand bidding method within an energy sharing game. In these studies, market participants engage in a (generalized) Nash game, and the resulting (generalized) Nash equilibrium is seen as their optimal bidding strategies. At this equilibrium, every participant maximizes their profits, and none has an incentive to deviate from their chosen strategies. To identify the optimal bidding strategies, the algorithms in \cite{mishra2022game}-\cite{chen2019energy} either necessitate a central coordinator to continuously gather and disseminate strategies or operate under the assumption that each participant can communicate with every other participant. However, these assumptions might not hold in certain real-world scenarios where central coordinators are absent, and aggregators can only communicate with their immediate neighbours. Consequently, there is a pressing need to develop a \textit{fully} distributed Nash equilibrium-seeking algorithm to find optimal bidding strategies for aggregators. 

There are some literature in the control community investigating fully distributed Nash equilibrium seeking algorithms under partial information scenarios. For example, various algorithms have been proposed for Nash games without coupling constraints \cite{tatarenko2020geometric, salehisadaghiani2019distributed} and with coupling constrains \cite{pavel2019distributed,bianchi2022fast}. However, these algorithms are complicated and inefficient, since they require each player to estimate the decisions of all the other players and increase the number of variables significantly. Furthermore, each player has to share its true decision to its neighbours. For aggregative games, \cite{koshal2016distributed} firstly proposed a distributed algorithm with diminishing step sizes, where each player only estimates the aggregate value of all decisions. Another works \cite{lei2018linearly,parise2019distributed} require an increasing number of communication rounds before each decision updates. To overcome the limitations of above mentioned algorithms, a project-gradient based distributed Nash equilibrium seeking algorithm for aggregative games was proposed in \cite{gadjov2020single}. The authors of \cite{bianchi2022fast} employed a proximal-point algorithm to enhance the convergence speed.  The convergence of these two algorithms can be guaranteed only providing they are initialized properly. In addition, they both require the so-called pseudo-gradient mappings to be strongly monotone, which restricts their direct implementation in our problem. 

In this work, we examine a scenario in which there's an energy supply deficit within a distribution network. The discrepancy between supply and demand is termed the "load adjustment requirement." The utility company encourages aggregators to modify their energy consumption and generation to fulfill this requirement. The main contributions of our study are as follows:
\begin{itemize}
\item We introduce a demand response scheme for energy balancing. In this scheme, aggregators bid strategically, and the utility company clears the market to meet the load adjustment requirement, taking into account both load adjustment capacities and power flow constraints. We frame the bidding behaviors of aggregators within this scheme as a network-constrained aggregative game and characterize their optimal bidding strategies.
\item Acknowledging the absence of coordinators and the fact that aggregators can only communicate with their neighbours, we present a \textit{fully} distributed generalized Nash equilibrium algorithm to pinpoint the optimal bidding strategies for aggregators in this game. In this method, aggregators share only an estimate of the aggregate value and a few auxiliary variables with their neighbours. Furthermore, our algorithm does not assume the pseudo-gradient mapping to be strongly monotone. We validate the effectiveness of our approach with a case study.
\end{itemize}
The paper is organized as follows: Section \uppercase\expandafter{\romannumeral2} introduces the demand response scheme for energy balancing in distribution network. Section \uppercase\expandafter{\romannumeral3} formulates the aggregative game among aggregators in this scheme. In Section \uppercase\expandafter{\romannumeral4}, a fully distributed Nash equilibrium seeking algorithm is proposed, and its convergence is formally proved. The effectiveness of the algorithm is verified in  Section
\uppercase\expandafter{\romannumeral5}. The paper closes with conclusions in Section \uppercase\expandafter{\romannumeral6}.

\subsection{Notation} Let $\mathbb{R}$ and $\mathbb{R}_+$ be the sets of real numbers and nonnegtive real numbers, respectively. $\mathbb{R}^n$ and $\mathbb{R}^{n\times m}$ denote the spaces of all $n$-dimension vectors and $n\times m$ matrices with real elements. We use $\mathbbb{1}$($\mathbbb{0}$) to denote the vector/matrix with all elements equal to 1(0) and use $I$ as the identity matrix.  We include the dimension of these vectors/matrices as a subscript, whenever needed. Given a set $\mathcal{N}=\{1,2,...,N\}$, $\col(x_n)_{n\in\mathcal{N}}$ denotes the stacked vector obtained from $x_n\in \mathbb{R}^{m_n}$,  $\diag(x_n)_{n\in\mathcal{N}}$ denotes the diagonal matrix with $x_n\in \mathbb{R}$ on its diagonal, $\blkdiag(A_n)_{n\in\mathcal{N}}$ denotes the block diagonal matrix with $A_n\in\mathbb{R}^{m_n\times l_n}$ as its diagonal blocks. The maximum (minimum) element in $\col(x_n)_{n\in\mathcal{N}}$ is denoted by $\max_{n\in \mathcal{N}} x_n$($\max_{i\in \mathcal{N}} x_n$). For vectors $x,y \in \mathbb{R}^n$ and a symmetric positive definite matrix $\Phi\in \mathbb{R}^{n\times n}$, 
$\langle   x,y\rangle_\Phi=\langle \Phi x,y\rangle$ denotes the $\Phi$-induced inner product,  $\|x\|_\Phi=\sqrt{\langle \Phi x,x\rangle}$ denotes the $\Phi$-induced norm, and we drop the index $\Phi$ for the case of standard norm/inner product $\Phi=I_n$. We use $\lambda_{\min}(\Phi)$, $\lambda_{\max}(\Phi)$, and $\Phi^{-1}$ to denote the minimum eigenvalue, maximum eigenvalue and inverse matrix of $\Phi$. For a matrix $A\in \mathbb{R}^{n\times m}$, we use $\|A\|$ to denote the maximum singular value of $A$. The Kronecker product is denoted by $\otimes$ and the Cartesian product of the sets $\Omega_n$, with $n\in \mathcal{N}$, by $\prod_{n\in\mathcal{N}}\Omega_n$.

\subsection{Operator theory}
We use $\mathrm{Id}(\cdot)$ to denote the identity operator. For a closed set $\Omega\in\mathbb{R}^n$, the mapping $\proj_{\Omega}:\mathbb{R}^n\rightarrow \Omega$ denotes the projection onto $\Omega$, i.e, $\proj_{\Omega}(x)=\argmin_{y\in\Omega}\|y-x\|$. The set-valued mapping $\Norcone_{\Omega}:\mathbb{R}^n\rightarrow \mathbb{R}^n$ denotes the normal cone operator for the set $\Omega\in \mathbb{R}^n$, i.e, $\Norcone_{\Omega}(x)=\emptyset$ if $x\notin\Omega$, and $\Norcone_{\Omega}(x)=\{v\in \mathbb{R}^n\mid \mathrm{sup}_{z\in\Omega}v^{\top}(z-x)\geq 0\}$ otherwise. A mapping $F:\Omega \rightarrow\mathbb{R}^n$ is $\ell$-Lipschitz continuous, with $\ell>0$, if $\|F(x)-F(y)\|\leq\ell\|x-y\|$ for all $x,y\in\Omega$. The mapping $F$ is $\mu$-strongly monotone, with $\mu>0$, if $(F(x)-F(y))^{\top}(x-y)\geq\mu\|x-y\|^2$ for all $x,y\in\Omega$. The mapping $F$ is $\eta$-averaged, with $\eta\in(0,1)$, if $\|F(x)-F(y)\|^2\leq \|x-y\|^2-\frac{1-\eta}{\eta}\|x-F(x)-(y-F(y))\|^2$, for all $x,y\in\Omega$. The mapping $F$ is $\beta$-cocoercive, with $\beta>0$, if $\beta\|F(x)-F(y)\|^2\leq (x-y)^\top (F(x)-F(y))$, for all $x,y\in\Omega$. The variational inequality problem $\mathrm{VI} \big(\Omega,F\big)$ is to find the point $\bar{x}\in \Omega$ such that $(x-\bar{x})^\top F(\bar{x})\geq 0$ for all $x\in\Omega$. We use ``$\circ$" to denote the composition of two mappings.

\section{Problem Statement}
\label{sec::prob}
This paper considers a distribution network consisting of a utility company and a set of aggregators $\mathcal{N}:=\{1,2,...,N\}$ with index $n\in\mathcal{N}$. Each aggregator is responsible for managing a group of prosumers equipped with flexible resources, including dispatchable generators and adjustable loads. The distribution network may experience an energy supply deficit in a certain time interval due to the upstream transmission network's prediction error of renewable generation outputs. The utility company can then incentivize the aggregators to adjust their pre-scheduled energy consumption or generation to offset this deficit. It broadcasts the total load adjustment requirement (the amount of energy deficit) $r\in\mathbb{R}_+$ to aggregators and determines how much of load adjustment $\{x_n\}_{n\in \mathcal{N}}$ should be provided by aggregator $n\in\mathcal{N}$ such that the energy balancing holds, namely
 \begin{equation}\label{eq_flexibility_balance}
    \sum_{n=1}^N x_n=r.
  \end{equation}

Motivated by the energy sharing mechanism proposed in \cite{chen2019energy}, we adopt the following demand response scheme.
\begin{itemize}
    \item Each aggregator $n\in\mathcal{N}$ submits its bid
    $$\beta_n\in \Omega_n:=\{\beta_n\mid \underline{\beta}\leq \beta_n\leq \bar{\beta}\}$$
    and its load adjustment capacity $\hat{x}_n$ to the utility company, where $\Omega_n$ is the feasible set, $\underline{\beta}$ and $\bar{\beta}$ are the minimum and maximum admissible bids. The bid $\beta_n$ indicates the level of willingness of the $n$th aggregator to adjust the load of its corresponding prosumers.

    \item The utility company clears the price $p$ as
\begin{equation}\label{eq_price}
    p=\frac{r-\mathbbb{1}^{\top}\beta}{\alpha N},
\end{equation}
where $\alpha\in\mathbb{R}^+$ is a constant imposed by the utility company. Note that the clearing price penalizes the mismatch between the amount of energy deficit and the total load adjustment bids with a factor of $1/(\alpha N)$.
The clearing load adjustment $x_n$ of aggregator $n$ is 
\begin{equation} \label{eq_load_adjustment}
  x_n=\frac{r-\mathbbb{1}^{\top}\beta}{N}+\beta_n, \,\forall\, n\in\mathcal{N},  
\end{equation}
or in a compact form as
\begin{equation}\label{eq_flexibility_allocation_compact}
    x=A\beta+c,
\end{equation}
where $\beta=\col(\beta_n)_{n\in\mathcal{N}}$, $x=\col(x_n)_{n\in\mathcal{N}}$, $A= I-\frac{1}{N} \mathbbb{1}\mathbbb{1}^{\top}$ and $c=\frac{r}{N}\mathbbb{1}$. 

Note that the load adjustment in \eqref{eq_load_adjustment}, equivalently \eqref{eq_flexibility_allocation_compact}, guarantees that the balancing condition \eqref{eq_flexibility_balance} holds
for any $\beta\in\Omega:=\prod\limits_{n\in \mathcal{N}}\Omega_n$.
\item The utility company verifies if the market clearing result satisfies the load adjustment capacity constraints 
\begin{equation} \label{eq_constraint_original}
    0 \le x_n \le \hat{x}_n, \forall\, n\in\mathcal{N},  
\end{equation}
and the power flow constraints \cite{chen2022energy} 
\begin{equation} \label{eq_power_flow}
    -\hat{f}_l \le \sum_{n=1}^N \pi_{ln}(e_n-x_n) \le \hat{f}_l, \, \forall \, l\in\mathcal{L},
\end{equation}
where $\hat{f}_l\in \mathbb{R}_+$ is the line capacity of line $l\in\mathcal{L}:=\{1,2,...,H\}$, $e_n$ is the pre-scheduled net load of aggregator $n$, and $\pi_{ln}\in \mathbb{R}$ is the line flow distribution factor from aggregator $n$ to line $l$.

\item If the constraints \eqref{eq_constraint_original} and \eqref{eq_power_flow} are not satisfied, the bids should be modified to the ``closest" feasible ones, that is, the solution to
\begin{equation} \label{modification}
    \begin{aligned}
        &\min_{\beta'\in \Omega} \, ||\beta'-\beta|| \\
        &\mathrm{s.t.} \, \eqref{eq_load_adjustment}, \eqref{eq_constraint_original}\, \mathrm{and} \, \eqref{eq_power_flow} \, \mathrm{hold}.
    \end{aligned}
\end{equation} 
\end{itemize}

\section{Game formulation}
Since the aggregators are rational and self-interested, they make a strategic bid in the demand response scheme. In this section, we formulate the competition among aggregators as an aggregative game and characterize their optimal bidding strategies.

Each aggregator $n\in\mathcal{N}$, as an intermediary, aims to maximize its net revenue of participating in demand response, that is, to minimize, 
\begin{equation} \label{eq_objective_original}
    J_n(x_n,p)=C_n(x_n)-p x_n,
\end{equation}
where $C_n(x_n)=q_n(x_n)x_n$ is the payment from aggregator $n$ to its prosumers with the pricing function $q_n(\cdot)$. For simplicity, we take this pricing function as $q_n(x_n)=a_n x_n+b_n$ with $a_n,b_n>0$. Note that the aggregators need to offer a higher price to secure more load adjustments from the prosumers.  

To show the explicit effect of the bidding strategies in the net revenue $J_n(\cdot, \cdot)$, we substitute $p$ and $x_n$ from \eqref{eq_price} and \eqref{eq_load_adjustment} in \eqref{eq_objective_original} and rewrite it as
\begin{multline}\label{eq_objective}
\bar J_n(\beta_n,\beta_{-n})=C_n\big((r-\mathbbb{1}^{\top}\beta)/N+\beta_n\big) \\
-(r-\mathbbb{1}^{\top}\beta+N\beta_n)(r-\mathbbb{1}^{\top}\beta)/(\alpha N^2),   
\end{multline}
where $\beta_{-n}=\col(\beta_m)_{m\in\mathcal{N} \setminus \{n\}}$.

Similarly, the constraints \eqref{eq_constraint_original} and \eqref{eq_power_flow} can also be written explicitly as constraints on the bids $\beta$, namely as
\[
-c\leq A\beta \leq \hat x-c,
\]
and
\[
-\hat{f}+\Pi(e-c) \leq \Pi A \beta \leq \hat{f} + \Pi (e-c),
\]
respectively, where $\hat{x}=\col(\hat{x}_n)_{n\in\mathcal{N}}$, $\hat{f}=\col(\hat{f}_l)_{n\in\mathcal{L}}$, $\Pi = [\pi_{ln}]_{l\in\mathcal{L},n\in\mathcal{N}}$. The latter two constraints can be written compactly as 
\begin{equation}\label{eq_global_set}
    \tilde{A}\beta\leq d,
\end{equation}
where 
\[
\tilde A= \begin{bmatrix}
  A \\
  -A\\
  -\Pi A \\
  \Pi A
 \end{bmatrix}, \,
 d= \begin{bmatrix}
  \hat{x}-c \\
  c\\
 \hat{f}-\Pi(e-c)\\
 \hat{f} + \Pi (e-c)
 \end{bmatrix}.
\]
Next, we rewrite $d$ as a summation of $N$ vectors, that is $d=\sum_{n=1}^N d_n$ with
\[
d_n=\frac{1}{N}\begin{bmatrix}
     -c \\
     c \\
    \hat{f}+\Pi c \\
     \hat{f}-\Pi c
\end{bmatrix} +
\begin{bmatrix}
     0 \\
     \vdots \\
     \hat{x}_n \\
     \vdots \\
     0
\end{bmatrix}+
\begin{bmatrix}
     0 \\
     \vdots \\
    -\Pi [0,...,e_n,...,0]^\top \\
     \Pi [0,...,e_n,...,0]^\top
\end{bmatrix}.
\]
Note that we define each $d_n$ such that the private information of each aggregator, namely $e_n$ and $\hat{x}_n$, are separated. 

Then, each aggregator faces the following constrained optimization problem 
\begin{equation} \label{prob_consumer}
    \begin{aligned}
        &\min_{\beta_n} \quad &&\bar J_n(\beta_n, \beta_{-n}) \\
        &\mathrm{s.t.} && \beta_n \in K_n(\beta_{-n}),
    \end{aligned}
\end{equation}
where 
\[
K_n(\beta_{-n}) =\{ \beta_n \in \Omega_n \mid \tilde{A}_n \beta_n \leq d-\sum\limits_{m\neq n}^N \tilde{A}_m\beta_m\},\]
and $\tilde{A}_n$ is the $n$th column of matrix $\tilde{A}$.

This can be viewed as a game among the aggregators, which can be written compactly as the triple: 
\begin{equation}\label{eq_game}
\mathcal{G} =\{\mathcal{N},K,\col(\bar J_n(\beta_n,\beta_{-n}))_{n\in \mathcal{N}}\}, 
\end{equation}
where $K=\prod\limits_{n\in \mathcal{N}}K_n(\beta_{-n})$ is the set of admissible strategies for all aggregators.

The game  $\mathcal{G}$ is a generalized Nash game (GNG) since their objective functions and the feasible strategy sets are both coupled. It is also an aggregative game since the objective function \eqref{eq_objective} are coupled only via the aggregative value of the bids. A point $\beta^* \in K$ is a generalized Nash equilibrium (GNE) of the game, if for all $n\in\mathcal{N}$, the following holds,
\begin{equation*}
 \bar J_n(\beta_n,\beta^*_{-n})\geq \bar J_n(\beta^*_n,\beta^*_{-n}),\ \ \forall \  \beta_n \in K_n(\beta^*_{-n}).
\end{equation*}

Based on the definition of GNE, each aggregator can minimize its objective at this point and none of them would unilaterally deviate from it. Hence, the GNE can be regarded as the optimal bidding strategies of all aggregators. In this manuscript, we focus on a specific subclass of GNE, namely v-GNE \cite{kulkarni2012variational}.  Specifically, each player in the game is penalized equally for deviating from coupling constraints at the v-GNE, which also corresponds to the solution of a variational inequality problem VI($K$,$F$), where $F$ is the pseudo-gradient mapping of the game defined as 
\begin{equation} \label{eq_pgm}
    F(\beta):=\col(f_n(\beta_n,\beta_{-n}))_{n \in \mathcal{N}},
\end{equation}
where 
\begin{equation*}
\begin{aligned}
f_n(\beta_n,\beta_{-n})&:=\frac{\partial}{\partial \beta_n}\bar J_n(\beta_n,\beta_{-n}) \\=&\frac{N-1}{N} C'_n(x_n)
+\frac{(\mathbbb{1}^{\top}\beta-r)(N-2)+N \beta_n}{\alpha N^2},
\end{aligned}    
\end{equation*}
with $C'_n(x_n)$ denoting the partial derivative of $C_n$ with respect to $x_n$.\footnote{Note that $C'(x_n)$ can also be stated in terms of the bids using \eqref{eq_load_adjustment}. However, we opted not to do so for the sake of readability of the expressions.}    
Note that the set $K$ is convex and compact. Furthermore, to satisfy Slater's condition, we assume that $K$ has at least one strictly feasible point. Then, the existence of v-GNE follows from \cite[Theorem~41(a)]{scutari2014real}. 

To prepare for the algorithm design in the next section, we also introduce a local variable $\sigma_n$ for each $n\in \mathcal{N}$ and define
\begin{equation}\label{eq_g}
\begin{aligned}
\hat f_n(\beta_n,\sigma_n)&:=\frac{N-1}{N}C'_n\left(\frac{r-N\sigma_n}{N}+\beta_n\right)\\
&+\frac{(N\sigma_n-r)(N-2)+N \beta_n}{\alpha N^2},
\end{aligned}
\end{equation}
and 
\begin{equation}\label{eq_hat_F}
    \hat{F}(\beta,{\sigma}):=\col(\hat f_n(\beta_n,\sigma_n))_{n\in\mathcal{N}},
\end{equation}
where $\sigma=\col(\sigma)_{n\in\mathcal{N}}$. 
Note that  $\hat{f_n}$ can be obtained from $f_n$ by replacing $\mathbbb{1}^{\top}\beta$ by $\sigma_n N$.
Hence,  $\hat{F}(\beta,\mathbbb{1}^{\top}\beta/N)=F(\beta)$.
The variable $\sigma_n$ will serve as a local estimate of the global quantity $\mathbbb{1}^{\top}\beta/N$ for the $n$th aggregator.   

\section{Algorithm design}
\subsection{Algorithm description}
In this section, we present our proposed algorithm to find the optimal bidding strategies (v-GNE) for aggregators. Motivated by the fact that there is no coordinators, and aggregators can only communicate with their neighbours, we devise a fully distributed v-GNE seeking protocol under partial information setting. We assume that aggregators communicate locally with their neighbours via a weighted communication graph $G$. The communication graph is assumed to be connected and undirected. Each aggregator $n\in \mathcal{N}$ maintains a local estimate $\sigma_n$ of the aggregative bid $\mathbbb{1}^\top \beta/N$ and local multiplier estimates $\lambda_n\in\mathbb{R}^M$ of the multipliers of coupling constraints \eqref{eq_global_set}. Two additional auxiliary variables $\psi_n\in\mathbb{R}$ and $z_n\in\mathbb{R}^M$ are communicated to the neighbouring aggregators with the aim of reaching the consensus on the local estimates $\sigma_n$'s and local multipliers $\lambda_n$'s. The proposed distributed protocol is given in Algorithm \ref{alg1}, where for all $n\in\mathcal{N}$, the step sizes $\tau_n,\upsilon_n,\rho_n,\delta_n,\eta_n$ and the parameter $\kappa$ are all positive, $w_{nm}$ is the weight of each link $\{n, m\}$ of the communication graph.

To write the algorithm more compactly, let
\[
\psi=\col(\psi_n)_{n\in\mathcal{N}}, \sigma=\col(\sigma_n)_{n\in\mathcal{N}}, z=\col(z_n)_{n\in\mathcal{N}},
\]
\[
\lambda=\col(\lambda_n)_{n\in\mathcal{N}}, \tau=\diag(\tau_n)_{n\in\mathcal{N}}, \upsilon=\diag(\upsilon_n)_{n\in\mathcal{N}},
\]
\[
\rho=\diag(\rho_n)_{n\in\mathcal{N}}, , \delta=\blkdiag(\delta_n\otimes I_M)_{n\in\mathcal{N}}, 
\]
\[
\bar{A}=\blkdiag(\tilde A_n)_{n\in\mathcal{N}}, \eta=\blkdiag(\eta_n\otimes I_M)_{n\in\mathcal{N}}. 
\]
and $\bar{d}=\col(d_n)_{n\in\mathcal{N}}$.
Consequently, we can write the dynamics in Algorithm \ref{alg1} as
\begin{equation}\label{eq_compact}
\begin{aligned}
&\beta^{k+1}=\proj_{\Omega}[\beta^k-\tau(  \hat F(\beta^k,\sigma^k)+\bar{A}^\top\lambda^{k})]\\
&\psi^{k+1}=\psi^k+\upsilon L_{\sigma}\sigma^k\\
&\sigma^{k+1}=\sigma^k+\rho(\kappa(\beta^k-\sigma^k)-L_{\sigma}(2\psi^{k+1}-\psi^k))\\
&z^{k+1}=z^k+\delta L_{\lambda}\lambda^k\\
&\lambda^{k+1}=\proj_{\mathbb{R}_{+}^{NM}}[\lambda^k-\eta(L_{\lambda}\lambda^k+\bar{d}-\bar{A}(2\beta^{k+1}-\beta^k)\\
&+L_{\lambda}(2z^{k+1}-z^k))]
\end{aligned}
\end{equation}
where $\hat F(\cdot, \cdot)$ is given by \eqref{eq_hat_F}, and $L_{\sigma}=L$, $L_{\lambda}=L\otimes I_{M}$, with $L$ denoting the Laplacian matrix of $G$.
\vspace{5pt}
\begin{algorithm}
\caption{Fully Distributed v-GNE Seeking Algorithm} \label{alg1}

\textbf{Initialization:} For each $n\in \mathcal{N}$, set $\beta_n^0 \in \Omega_n$, $\sigma_n^0\in \mathbb{R}$, $\psi_n^0 \in \mathbb{R}$, $z_n \in \mathbb{R}^M$, $\lambda_n \in \mathbb{R}_{+}^M$.
\par\textbf{Iterate until convergence:} 
\par\textit{Communication at the $k$th step:} Each aggregator $n\in \mathcal{N}$ communicate $\sigma_n^k$, $\psi_n^k$, $z_n^k$, $\lambda_n^k$ to its neighbouring aggregators $m\in\mathcal{N}_n$.
\par\textit{Local variable update at the $k$th step:}
\small
\begin{align*}
\beta_n^{k+1}&=\proj_{\Omega_n}[\beta_n^k-\tau_n(\hat{f}_n(\beta_n^k,\sigma_n^k)+\tilde{A}_n^\top\lambda_n^k)]\\
\psi_n^{k+1}&=\psi_n^k+\upsilon_n\sum_{m\in\mathcal{N}_n}w_{nm}(\sigma_n^k-\sigma_m^k) \\
\sigma_n^{k+1}&=\sigma^k_n+\rho_n\big(\kappa(\beta^k_n-\sigma_n^k)\\
&-\sum_{m\in\mathcal{N}_n}w_{mn}(2(\psi_n^{k+1}-\psi_m^{k+1})-(\psi_n^k-\psi_m^k))\big)\\
z_n^{k+1}&=z_n^k+\delta_n\sum_{m\in\mathcal{N}_n}w_{nm}(\lambda_n^k-\lambda_m^k) \\
\lambda_n^{k+1}&=\proj_{\mathbb{R}_{+}^M}[\lambda_n^k-\eta_n\big(\sum_{m\in\mathcal{N}_n}w_{nm}(\lambda_n^k-\lambda_m^k)+d_n+\tilde{A}_n(\beta_n^k\\
&-2\beta_n^{k+1})+\sum_{m\in\mathcal{N}_n}w_{nm}(2(z_n^{k+1}-z_m^{k+1})-(z_n^k-z_m^k))\big)]
\end{align*}
\end{algorithm}

\subsection{Steady-state analysis}
Before providing the convergence analysis, we show that the steady state of the dynamics in \eqref{eq_compact} yields the v-GNE of the game. In what follows, we show that \eqref{eq_compact} can be further written as the following preconditioned forward-backward iteration,
\begin{equation}\label{eq_iteration}
\omega^{k+1}=\mathcal{V}_{\Phi}\circ\mathcal{U}_{\Phi}(\omega^k),
\end{equation}
where $\omega\in\Omega\times \mathbb{R}^{2N(1+M)}:=\col(\beta,\psi,\sigma,z,\lambda)$, $\mathcal{V}_{\Phi}:=(\mathrm{Id}+\Phi^{-1}\mathcal{B})^{-1}$, $\mathcal{U}_{\Phi}:=(\mathrm{Id}-\Phi^{-1}\mathcal{A})$. The mappings $\mathcal{A}$ and $\mathcal{B}$ and the preconditioned matrix $\Phi$ are defined as
\begin{equation}
\mathcal{A}\\:=\begin{bmatrix}
\hat F(\beta,\sigma)\\
0\\
\kappa(\sigma-\beta)\\
0\\
\bar d+L_{\lambda}\lambda
\end{bmatrix},
\mathcal{B}\\:=\begin{bmatrix}
\Norcone_{\Omega}(\beta)+\bar{A}^\top\lambda\\
-L_{\sigma}\sigma\\
L_{\sigma}\psi\\
-L_{\lambda}\lambda\\
\Norcone_{\mathbb{R}_{+}^M}(\lambda)-\bar{A}x+L_{\lambda}z
\end{bmatrix},
\end{equation}
\begin{equation}
\Phi:=\begin{bmatrix}
\tau^{-1} & 0 & 0 & 0 & -\bar{A}^\top\\
0 & \upsilon^{-1} & L_{\sigma} & 0 & 0\\
0 & L_{\sigma} & \rho^{-1} & 0 & 0\\
0 & 0 & 0 & \delta^{-1} & L_{\lambda}\\
-\bar{A} & 0 & 0 & L_{\lambda} & \eta^{-1}
\end{bmatrix}.
\end{equation}

The main result of this subsection is provided below.
\begin{theorem}\label{proposition_steady}
Assume that the matrix $\Phi$ is positive definite. Then, the dynamics \eqref{eq_compact} is equivalent to the forward-backward iteration \eqref{eq_iteration}; in particular, the steady state $\omega^*=(\beta^*,\psi^*,\sigma^*,z^*,\lambda^*)$ of \eqref{eq_compact} coincides with a fixed point of iteration \eqref{eq_iteration} and a zero of the mapping $\mathcal{A}+\mathcal{B}$. Moreover, $\beta^*$ is a v-GNE of the game $\mathcal{G}$.   \end{theorem}
\begin{proof}
See Appendix.
\end{proof}
\subsection{Convergence analysis}
As we observed, the steady-state of the proposed algorithm coincides with a v-GNE of the game. Next, under suitable choices of step sizes, we show that the proposed algorithm converges to this point, as desired. To this end, we first require a few technical results.

The following lemma establishes a cocoercivity property.
\begin{lemma}\label{lemma_tilde_A}
Let $\tilde{\mathcal{A}}$ be defined as
\begin{equation}\label{eq_tidle_A}
\tilde{\mathcal{A}}\\:\begin{bmatrix}
\beta \\
\sigma
\end{bmatrix} \rightarrow  \begin{bmatrix}
\hat F(\beta,\sigma)\\
\kappa (\sigma-\beta)
\end{bmatrix}.
\end{equation}
Let $\mu_n:=2a_n\frac{N-1}{N}+\frac{1}{\alpha N}$ and $\ell_n:=-2a_n\frac{N-1}{N}+\frac{N-2}{\alpha N}$, for each $n\in \mathcal{N}$. 
Assume that \footnote{The assumption is satisfied if the parameters of the cost function, namely $\{a_n\}_{n\in \mathcal{N}}$ are sufficiently uniform.} 
\begin{equation}
\sqrt{\max_{n\in\mathcal{N}} \mu_n} - \sqrt{\min_{n\in\mathcal{N}} \mu_n} \le 2\gamma,
\end{equation}
with $\gamma=\sqrt{\frac{N-1}{\alpha N}}$
and let $\kappa$ be chosen as
\begin{equation} \label{eq_k}
\begin{aligned}
    \kappa \in \left(\sqrt{\max_{n\in\mathcal{N}} \mu_n}-\gamma, \sqrt{\min_{n\in\mathcal{N}} \mu_n}+\gamma\right).
\end{aligned}
\end{equation}
Then, the mapping $\tilde{\mathcal{A}}$ is $\tilde\epsilon$-cocoercive with $\tilde \epsilon=\frac{\min_{n\in\mathcal{N}}\bar\epsilon_n}{\max_{n\in\mathcal{N}}\underline \epsilon_n}$, where
\begin{equation*}
 \begin{aligned}
\bar\epsilon_n&=-\sqrt{(\mu_n-\kappa)^2+(\ell_n-\kappa)^2}+\kappa+\mu_n,\\
   \underline \epsilon_n &=  \mu^2_n+\ell^2_n+2\kappa^2\\ &\quad + \sqrt{(\mu_n+\ell_n)^2(\mu_n-\ell_n)^2+4(\kappa^2-\mu_n\ell_n)^2}.
\end{aligned}   
\end{equation*}
\end{lemma}
\begin{proof}
See Appendix.
\end{proof}

Next, we show the mapping $\mathcal{V}_{\Phi}\circ\mathcal{U}_{\Phi}$ is averaged if the step sizes are chosen small enough.

\begin{lemma}\label{lemma_positive2}
The forward-backward iteration in \eqref{eq_iteration}, is $\theta$-averaged, with $\theta=\frac{1}{2-1/(2\xi)}\in(0,1)$, if $\kappa$ satisfies \eqref{eq_k} and for all $n\in\mathcal{N}$, 
\begin{equation}\label{eq_stepsize1}
 \tau_n<2\epsilon,\, \upsilon_n<2\epsilon. \, \delta_n > 2\epsilon,   
\end{equation}
\begin{equation}\label{eq_stepsize2}
 \rho_n^{-1}>\lambda_{\max}^2(L)\left(\frac{1}{\max_{n\in\mathcal{N}}\upsilon_n}-\frac{1}{2\epsilon}\right)^{-1}+\frac{1}{2\epsilon},   
\end{equation}
\begin{equation}\label{eq_stepsize3}
 \begin{aligned}
    \eta_n^{-1}&>\|\bar{A}\|^2\left(\frac{1}{\max_{n\in\mathcal{N}}{\tau_n}}-\frac{1}{2\epsilon}\right)^{-1}\\&+\lambda_{\max}^2(L)\left(\frac{1}{\max_{n\in\mathcal{N}}{\delta_n}}-\frac{1}{2\epsilon}\right)^{-1}+\frac{1}{2\epsilon},
\end{aligned}   
\end{equation}
where $\xi=\frac{\epsilon}{\lambda_{\max}(\Phi^{-1})}$,  $\epsilon=\min\{\tilde\epsilon,1/\lambda_{\max}(L)\}$ and $\tilde\epsilon$ is given by Lemma \ref{lemma_tilde_A}. 
\end{lemma}
\begin{proof}
See Appendix.
\end{proof}

Note that, by the proof of Lemma \ref{lemma_positive2}, the conditions \eqref{eq_stepsize1}, \eqref{eq_stepsize2}, \eqref{eq_stepsize3}, guarantee positive-definiteness of $\Phi$ that was assumed in Theorem \ref{proposition_steady}.
Now, we are ready to state the main result concerning the convergence of the algorithm to the v-GNE of the game.
\begin{theorem}\label{lemma_convegence}
Suppose  $\kappa$ satisfies \eqref{eq_k} and let the step sizes be chosen as in Lemma \ref{lemma_positive2}. Then, the solutions of the algorithm \ref{alg1} converge to the zero of the mapping $\mathcal{A}+\mathcal{B}$; in particular, $\beta$ converges to $\beta^*$, the v-GNE of game $G$.
\end{theorem}
\begin{proof}
 It follows by Theorem \ref{proposition_steady} that Algorithm \ref{alg1} corresponds to the iteration \eqref{eq_iteration} of the mapping $\mathcal{V}_{\Phi}\circ\mathcal{U}_{\Phi}$. By Lemma \ref{lemma_positive2}, this mapping is $\theta$-averaged, with $\theta \in(0,1)$. Then the sequence generated by the iteration \eqref{eq_iteration} converges to $\omega^*=(\beta^*,\psi^*,\sigma^*,z^*,\lambda^*)$, i.e, the zero of the mapping $\mathcal{A}+\mathcal{B}$ by \cite[Proposition 5.15(iii)]{bauschke2011convex}. In particular, $\beta^*$ is the v-GNE of game $\mathcal{G}$ \eqref{eq_game}. 
\end{proof}
\section{Case study}
We perform the numerical study on the modified IEEE 33 bus distribution network with five areas, as shown in Fig.~\ref{fig:case}. The prosumers of each area is managed by an aggregator; see the red numbers $1$ to $5$ in the figure. The slack bus 1 connects the distribution network to the upstream network and experiences an energy supply deficit. The five aggregators are physically connected by four solid lines $(3,19)$, $(4,5)$, $(7,26)$ and $(9,10)$. Note that aggregator $2$ can produce energy with the dispatchable generator. They also can communicate with their neighbours through dash lines with different weights. We choose parameters $r=600\mathrm{kWh}$, $\alpha=1$, $\beta_{\min}=0$, $\beta_{\max}=150$ and the power flow limits for the four lines are $ \hat f = [1.40,6.0,2.0,2.0]\times 1000 \mathrm{kWh}$. The  aggregators' data is shown in Table ~\ref{table_data}.

\begin{figure}[ht]
\begin{center}
\includegraphics[width=0.48\textwidth]{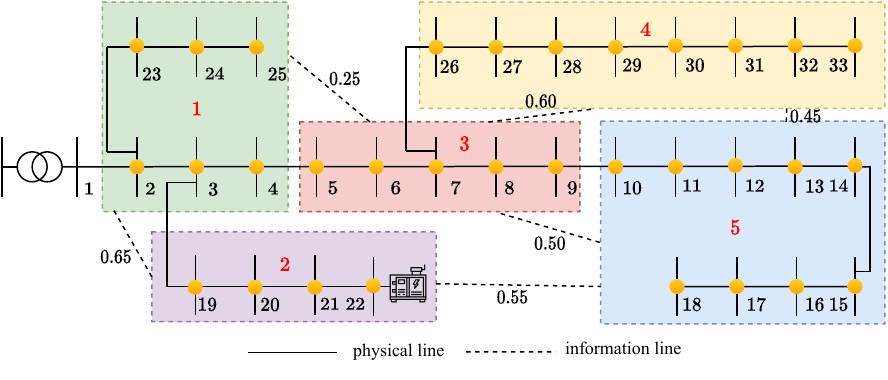}
\end{center}
\caption{Physical network and communication network among aggragators}
\label{fig:case}
\end{figure}

\begin{table}[ht]
\caption{Simulation data}
\vspace{-8pt}
\begin{center}
\begin{tabular}{ M{1cm} M{1.6cm} M{1.2cm} M{1.2cm} M{1cm} }
\toprule
Aggregator & $a_n (\$/(\mathrm{kWh})^2)$ & $b_n (\$/\mathrm{kWh})$ & $e_n(\mathrm{kWh})$ & $\hat{x}_n(\mathrm{kWh})$\\ 
\midrule
 1 & 0.0050 & 0.40 & 1250 & 250\\
2 & 0.0065 & 0.38 & -1300 & 200 \\
3 & 0.0085 & 0.36 & 1050 & 250 \\
 4 & 0.0070 & 0.37 & 1700 & 110 \\
 5 & 0.0095 & 0.80 & 1480 & 220 \\
\bottomrule
\end{tabular}
\end{center}

\label{table_data}
\end{table}

Fig.~\ref{fig:beta} depicts the evolution of bids $\beta$ and the estimates $\sigma$ of the aggregative bid $\mathbbb{1}^{\top}\beta/N$ in Algorithm \ref{alg1}. It can be seen that both the bids and the estimates can convergence in 600 iterations. Furthermore, each aggregator can finally estimate the true aggregative bid. We also show the evolution of multiplier estimates $\lambda$ and load adjustment $x$ in Fig.~\ref{fig:x}. The local multiplier estimates also convergence to the same values for each coupling constraint. The two groups of positive multipliers means there are two constraints affecting the optimal bidding strategies of aggregators. Based on the simulation setting, the abilities of aggregator $2$ and $3$ to provide load adjustment are restricted to the capacity of line $(3,19)$ and aggregator $3$'s load adjustment capacity, respectively.

\begin{figure}[ht]
\begin{center}
\includegraphics[width=0.48\textwidth]{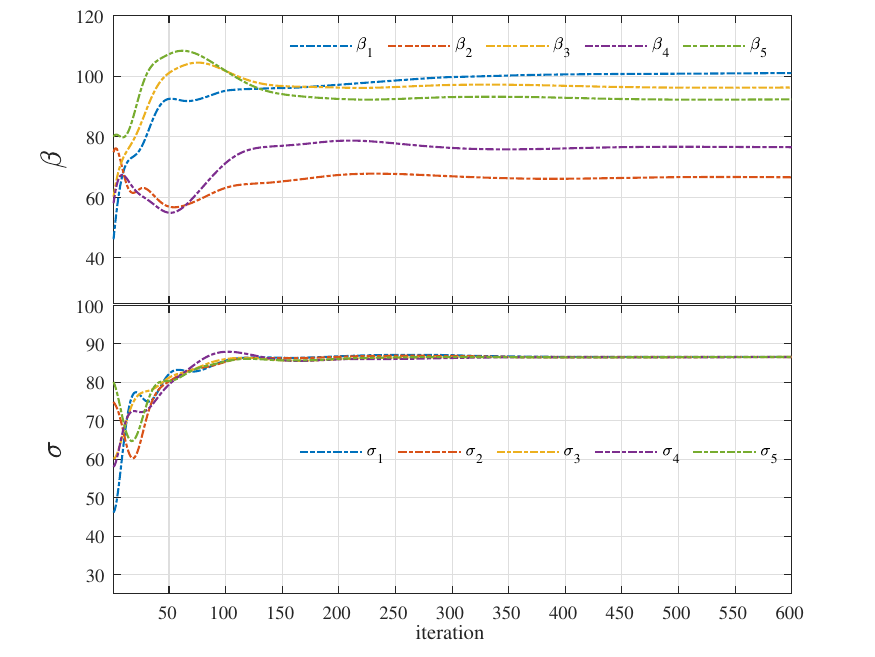}
\end{center}
\vspace{-10pt}
\caption{The evolution of $\beta$ and $\sigma$}
\label{fig:beta}
\end{figure}

\begin{figure}[ht]
\begin{center}
\includegraphics[width=0.48\textwidth]{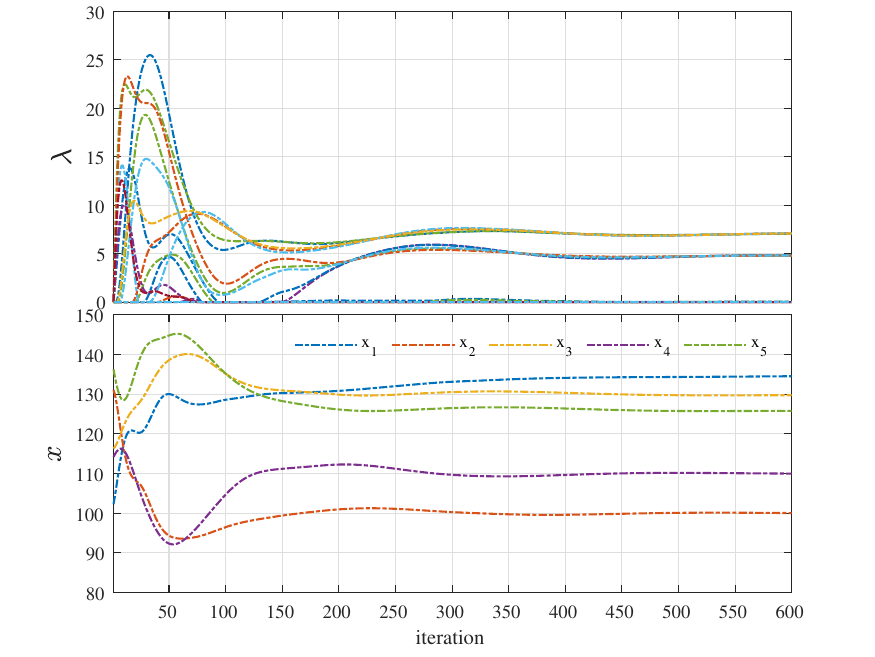}
\end{center}
\vspace{-10pt}
\caption{The evolution of $\lambda$ and $x$}
\label{fig:x}
\end{figure}

\section{Conclusion}
Aggregators can bid strategically in the demand response scheme to meet the load adjustment requirement. The bidding behaviours among them can be modeled as a network-constrained aggregative game. To find the optimal bidding strategies in a partial information setting, we propose a fully distributed Nash equilibrium seeking algorithm and give the upper bounds of the fixed step sizes. Numerical study show the effectiveness of this algorithm. Considering general cost functions and more precise physical network models are of interest for future
research. 

\bibliographystyle{IEEEtran} 
\bibliography{ref}

\begin{thebibliography}{10}
\providecommand{\url}[1]{#1}
\csname url@samestyle\endcsname
\providecommand{\newblock}{\relax}
\providecommand{\bibinfo}[2]{#2}
\providecommand{\BIBentrySTDinterwordspacing}{\spaceskip=0pt\relax}
\providecommand{\BIBentryALTinterwordstretchfactor}{4}
\providecommand{\BIBentryALTinterwordspacing}{\spaceskip=\fontdimen2\font plus
\BIBentryALTinterwordstretchfactor\fontdimen3\font minus \fontdimen4\font\relax}
\providecommand{\BIBforeignlanguage}[2]{{%
\expandafter\ifx\csname l@#1\endcsname\relax
\typeout{** WARNING: IEEEtran.bst: No hyphenation pattern has been}%
\typeout{** loaded for the language `#1'. Using the pattern for}%
\typeout{** the default language instead.}%
\else
\language=\csname l@#1\endcsname
\fi
#2}}
\providecommand{\BIBdecl}{\relax}
\BIBdecl

\bibitem{sinsel2020challenges}
S.~R. Sinsel, R.~L. Riemke, and V.~H. Hoffmann, ``Challenges and solution technologies for the integration of variable renewable energy sources—a review,'' \emph{renewable energy}, vol. 145, pp. 2271--2285, 2020.

\bibitem{gkatzikis2013role}
L.~Gkatzikis, I.~Koutsopoulos, and T.~Salonidis, ``The role of aggregators in smart grid demand response markets,'' \emph{IEEE Journal on selected areas in communications}, vol.~31, no.~7, pp. 1247--1257, 2013.

\bibitem{hupez2022pricing}
M.~Hupez, J.-F. Toubeau, I.~Atzeni, Z.~De~Gr{\`e}ve, and F.~Vall{\'e}e, ``Pricing electricity in residential communities using game-theoretical billings,'' \emph{IEEE Transactions on Smart Grid}, vol.~14, no.~2, pp. 1621--1631, 2022.

\bibitem{mishra2022game}
M.~K. Mishra and S.~Parida, ``A game theoretic horizon decomposition approach for real-time demand-side management,'' \emph{IEEE Transactions on Smart Grid}, vol.~13, no.~5, pp. 3532--3545, 2022.

\bibitem{scarabaggio2022noncooperative}
P.~Scarabaggio, R.~Carli, and M.~Dotoli, ``Noncooperative equilibrium-seeking in distributed energy systems under ac power flow nonlinear constraints,'' \emph{IEEE Transactions on Control of Network Systems}, vol.~9, no.~4, pp. 1731--1742, 2022.

\bibitem{chen2019energy}
Y.~Chen, S.~Mei, F.~Zhou, S.~H. Low, W.~Wei, and F.~Liu, ``An energy sharing game with generalized demand bidding: Model and properties,'' \emph{IEEE Transactions on Smart Grid}, vol.~11, no.~3, pp. 2055--2066, 2019.

\bibitem{tatarenko2020geometric}
T.~Tatarenko, W.~Shi, and A.~Nedi{\'c}, ``Geometric convergence of gradient play algorithms for distributed nash equilibrium seeking,'' \emph{IEEE Transactions on Automatic Control}, vol.~66, no.~11, pp. 5342--5353, 2020.

\bibitem{salehisadaghiani2019distributed}
F.~Salehisadaghiani, W.~Shi, and L.~Pavel, ``Distributed nash equilibrium seeking under partial-decision information via the alternating direction method of multipliers,'' \emph{Automatica}, vol. 103, pp. 27--35, 2019.

\bibitem{pavel2019distributed}
L.~Pavel, ``Distributed gne seeking under partial-decision information over networks via a doubly-augmented operator splitting approach,'' \emph{IEEE Transactions on Automatic Control}, vol.~65, no.~4, pp. 1584--1597, 2019.

\bibitem{bianchi2022fast}
M.~Bianchi, G.~Belgioioso, and S.~Grammatico, ``Fast generalized nash equilibrium seeking under partial-decision information,'' \emph{Automatica}, vol. 136, p. 110080, 2022.

\bibitem{koshal2016distributed}
J.~Koshal, A.~Nedi{\'c}, and U.~V. Shanbhag, ``Distributed algorithms for aggregative games on graphs,'' \emph{Operations Research}, vol.~64, no.~3, pp. 680--704, 2016.

\bibitem{lei2018linearly}
J.~Lei and U.~V. Shanbhag, ``Linearly convergent variable sample-size schemes for stochastic nash games: Best-response schemes and distributed gradient-response schemes,'' in \emph{2018 IEEE Conference on Decision and Control (CDC)}.\hskip 1em plus 0.5em minus 0.4em\relax IEEE, 2018, pp. 3547--3552.

\bibitem{parise2019distributed}
F.~Parise, B.~Gentile, and J.~Lygeros, ``A distributed algorithm for almost-nash equilibria of average aggregative games with coupling constraints,'' \emph{IEEE Transactions on Control of Network Systems}, vol.~7, no.~2, pp. 770--782, 2019.

\bibitem{gadjov2020single}
D.~Gadjov and L.~Pavel, ``Single-timescale distributed gne seeking for aggregative games over networks via forward--backward operator splitting,'' \emph{IEEE Transactions on Automatic Control}, vol.~66, no.~7, pp. 3259--3266, 2020.

\bibitem{chen2022energy}
Y.~Chen, C.~Zhao, S.~H. Low, and A.~Wierman, ``An energy sharing mechanism considering network constraints and market power limitation,'' \emph{IEEE Transactions on Smart Grid}, vol.~14, no.~2, pp. 1027--1041, 2022.

\bibitem{kulkarni2012variational}
A.~A. Kulkarni and U.~V. Shanbhag, ``On the variational equilibrium as a refinement of the generalized nash equilibrium,'' \emph{Automatica}, vol.~48, no.~1, pp. 45--55, 2012.

\bibitem{scutari2014real}
G.~Scutari, F.~Facchinei, J.-S. Pang, and D.~P. Palomar, ``Real and complex monotone communication games,'' \emph{IEEE Transactions on Information Theory}, vol.~60, no.~7, pp. 4197--4231, 2014.

\bibitem{bauschke2011convex}
H.~H. Bauschke, P.~L. Combettes \emph{et~al.}, \emph{Convex analysis and monotone operator theory in Hilbert spaces}.\hskip 1em plus 0.5em minus 0.4em\relax Springer, 2011, vol. 408.

\bibitem{combettes2014variable}
P.~L. Combettes and B.~C. V{\~u}, ``Variable metric forward--backward splitting with applications to monotone inclusions in duality,'' \emph{Optimization}, vol.~63, no.~9, pp. 1289--1318, 2014.

\bibitem{facchinei2003finite}
F.~Facchinei and J.-S. Pang, \emph{Finite-dimensional variational inequalities and complementarity problems}.\hskip 1em plus 0.5em minus 0.4em\relax Springer, 2003.

\bibitem{combettes2015compositions}
P.~L. Combettes and I.~Yamada, ``Compositions and convex combinations of averaged nonexpansive operators,'' \emph{Journal of Mathematical Analysis and Applications}, vol. 425, no.~1, pp. 55--70, 2015.

\end{thebibliography}

\appendix

\begin{proof}[Proof of Theorem \ref{proposition_steady}]
The techniques in the proof are inspired by \cite{gadjov2020single}.
Firstly, we show the dynamics \eqref{eq_compact} is equivalent to
\begin{equation}\label{eq_iteration1}
0 \in \mathcal{A}(\omega^k)+\mathcal{B}(\omega^{k+1})+\Phi(\omega^{k+1}-\omega^k).
\end{equation}
To do it, we plug the expressions of $\mathcal{A}$,  $\mathcal{B}$ and $\Phi$ into \eqref{eq_iteration1} to obtain
\begin{align*}
0 & \in \begin{bmatrix}
F(\beta^k,\sigma^k)\\
0\\
\kappa(\sigma^k-\beta^k)\\
0\\
\hat{b}+L_{\lambda}\lambda^k
\end{bmatrix}+
\begin{bmatrix}
\Norcone_{\Omega}(\beta^{k+1})+\bar{A}^\top\lambda^{k+1}\\
-L_{\sigma}\sigma^{k+1}\\
L_{\sigma}\psi^{k+1}\\
-L_{\lambda}\lambda^{k+1}\\
\Norcone_{\mathbb{R}_{+}^{NM}}(\lambda^{k+1})-\bar{A}\beta^{k+1}+L_{\lambda}z^{k+1}
\end{bmatrix}  \\
& + \begin{bmatrix}
    \tau^{-1}(\beta^{k+1}-\beta^k)-\bar{A}^\top(\lambda^{k+1}-\lambda^k) \\ v^{-1}(\psi^{k+1}-\psi^k)+L_{\sigma}(\sigma^{k+1}-\sigma^k) \\
    L_{\sigma}(\psi^{k+1}-\psi^k) +v^{-1}(\sigma^{k+1}-\sigma^k) \\
    \delta^{-1}(z^{k+1}-z^k)+L_{\lambda}(\lambda^{k+1}-\lambda^k) \\
    -\bar{A}^\top(\beta^{k+1}-\beta^k)+L_{\lambda}(z^{k+1}-z^k)+\eta^{-1} (\lambda^{k+1}-\lambda^k)
\end{bmatrix}.
\end{align*}

By reformulating the second, third and fourth inclusion as
\begin{align*}
v^{-1}(\psi^{k+1}-\psi^k)-L_{\sigma}\sigma^k=0,  
\end{align*}
\begin{align*}
    \kappa(\sigma^k-\beta^k)- L_{\sigma}\psi^k +v^{-1}(\sigma^{k+1}-\sigma^k)=0,
\end{align*}
\begin{align*}
\delta^{-1}(z^{k+1}-z^k)-L_{\lambda}\lambda^k=0,  
\end{align*}
It is easily observed that they are equivalent to the second, third and fourth equality in \eqref{eq_compact}.

The first inclusion can be written as
\begin{align*}
    0\in F(\beta^k,\sigma^k)+\Norcone_{\Omega}(\beta^{k+1})+\tau^{-1}(\beta^{k+1}-\beta^k)+\bar{A}^\top\lambda^k,
\end{align*}
and thus as
\begin{align*}
\beta^k-\tau(F(\beta^k,\sigma^k)+\bar{A}^\top\lambda^k) \in \tau\Norcone_{\Omega}(\beta^{k+1})+\beta^{k+1}.
\end{align*}
Then, the first inequality in \eqref{eq_compact} can be obtained by using $\proj_{\Omega}=(\mathrm{Id}+\Norcone_{\Omega})^{-1}$ (Example 23.4 in Reference \cite{bauschke2011convex}) and  ${\tau}^{-1}\Norcone_{\Omega}(\beta)=\Norcone_{\Omega}(\beta)$. 
We note that the latter holds since for any  $y\in \Norcone_{\Omega}(\beta)$ and $\tau_n>0$, we have $\mathrm{sup}_{z\in\Omega}\sum_{n=1}^N \tau^{-1}_n y_n(z_n-\beta_n)\geq 0$, i.e., $\tau^{-1}y\in \Norcone_{\Omega}(\beta)$.

Similarly, the last inclusion is written as
\begin{align*}
     0&\in \hat{b}+L_{\lambda}\lambda^k+\Norcone_{\mathbb{R}_{+}^{NM}}(\lambda^{k+1})-\bar{A}\beta^{k+1}+L_{\lambda}z^{k+1} \\
     -& \bar{A}^\top(\beta^{k+1}-\beta^k)+L_{\lambda}(z^{k+1}-z^k)+\eta^{-1} (\lambda^{k+1}-\lambda^k),
\end{align*}
 and thus as
\begin{align*}
     &-\bar{d}-L_{\lambda}\lambda^k-\bar{A}(2\beta^{k+1}-\beta^k)+L_{\lambda}(2z^{k+1}-z^k)+\eta^{-1} \lambda^k \\
     &\in  \Norcone_{\mathbb{R}_{+}^{NM}}(\lambda^{k+1})+\eta^{-1} \lambda^{k+1}
\end{align*}
which gives the last equality in \eqref{eq_compact}.

Since $\Phi$ is positive definite, \eqref{eq_iteration1} can be written as $(\mathrm{Id}-\Phi^{-1}\mathcal{A})({\omega}^k) \in (\mathrm{Id}+\Phi^{-1}\mathcal{B})({\omega}^{k+1})$. 

Next, we show that $(\mathrm{Id}+\Phi^{-1}\mathcal{B})^{-1}$ is singled-valued. The mapping $\mathcal{B}$ can be split further as
\begin{equation}
\begin{aligned}
\mathcal{B}:=&\mathcal{B}_1+\mathcal{B}_2:=\begin{bmatrix}
\Norcone_{\Omega}(\beta)\\
0\\
0\\
0\\
\Norcone_{\mathbb{R}_{+}^{NM}}(\lambda)
\end{bmatrix}\\
&+
\begin{bmatrix}
0&0&0&0&\bar{A}^\top\\
0&0&-L_{\sigma}&0&0\\
0&L_{\sigma}&0&0&0\\
0&0&0&0&-L_{\lambda}\\
-\bar{A}&0&0&L_{\lambda}&0
\end{bmatrix}
\begin{bmatrix}
\beta\\ \psi\\ \sigma \\ z \\ \lambda
\end{bmatrix}.
\end{aligned}
\end{equation}
The mappings $\mathcal{B}_1$ is maximally monotone since normal cones of closed convex sets are maximally monotone (Example 20.41 in Reference \cite{bauschke2011convex}), and the concatenation preserves  maximality (Proposition 20.23 in Reference \cite{bauschke2011convex}). We also deduce from Example 20.30 in Reference \cite{bauschke2011convex} that the mapping $\mathcal{B}_2$ is maximally monotone since it is linear and skew-symmetric, i.e, $\mathcal{B}_2^{\top}=-\mathcal{B}_2$. Then, the maximal monotonicity of $\mathcal{B}$ follows from Corollary 24.4 in Reference \cite{bauschke2011convex}. Since $\Phi$ is positive definite and $\mathcal{B}$ is maximally monotone, we obtain from Lemma 3.7 in Reference \cite{combettes2014variable} that the mapping $\Phi^{-1}\mathcal{B}$ is maximally monotone in the $\Phi$-induced norm, which further implies that  $\mathcal{V}_{\Phi}=(\text{Id}+\Phi^{-1}\mathcal{B})^{-1}$ is firmly nonexpansive (Proposition 23.7 in Reference \cite{bauschke2011convex}). It then follows from Remark 4.24(iii)] in Reference \cite{bauschke2011convex} that $\mathcal{V}_{\Phi}$ is $\frac{1}{2}$-averaged, and in turn, singled-valued. Consequently \eqref{eq_iteration1} can further be written as \eqref{eq_iteration}.

Let $\omega^*=(\beta^*,\psi^*,\sigma^*,z^*,\lambda^*)$ be the steady state of dynamics \eqref{eq_compact}, hence is a fixed point of \eqref{eq_iteration}. By continuity, the following equivalences hold, $\omega^*=(\mathrm{Id}+\Phi^{-1}\mathcal{B})^{-1}\circ(\mathrm{Id}-\Phi^{-1}\mathcal{A})(\omega^*)\Leftrightarrow (\mathrm{Id}-\Phi^{-1}\mathcal{A})(\omega^*)\in (\mathrm{Id}+\Phi^{-1}\mathcal{B})(\omega^*)\Leftrightarrow 0\in (\mathcal{A}+\mathcal{B})(\omega^*)$, thus, it follows,
\begin{equation}\label{eq_zero}
\begin{aligned}
0&\in \Norcone_{\Omega}(\beta^*)+\hat{F}(\beta^*,\sigma^*)+\bar{A}^\top\lambda^*\\
0&=L_{\sigma}\sigma^* \\
0&=\kappa(\sigma^*-\beta^*)+L_{\sigma}\psi^*\\
0&=L_{\lambda}\lambda^* \\
0&\in \Norcone_{\mathbb{R}_{+}^{NM}}(\lambda^*)+L_{\lambda}\lambda^*-\bar{A}\beta^*+L_{\lambda}z^*-\bar{d}
\end{aligned}
\end{equation}
From the second and forth equality, we have $\sigma^*=\gamma_1\mathbbb{1}_N$ for some $\gamma_1\in\mathbb{R}$ and $\lambda^*=\mathbbb{1}_N\otimes\gamma_2$ for some $\gamma_2\in\mathbb{R}^M$. Then, the third equality becomes $0 = \gamma_1\mathbbb{1}_N-\beta^*+\kappa^{-1}L_{\sigma}\psi^*$. Left-multiplying both sides of the above equality by $\mathbbb{1}^{\top}_N$ gives $\gamma_1=\frac{1}{N}\mathbbb{1}^{\top}_N \beta^*$, which means that $\hat{F}(\beta^*,\sigma^*)=F(\beta^*)$. The first inclusion is subsequently written as 
\begin{equation}
  0 \in F(\beta^*)+\Norcone_{\Omega}(\beta^*)+\tilde A^\top \gamma_2.  
\end{equation}
We left-multiply both sides of the last equality in \eqref{eq_zero} by $\mathbbb{1}^{\top}_N\otimes I_M$ yields
\begin{equation}
  0 \in \Norcone_{\mathbb{R}_+^M}(\gamma_2)-\tilde A \beta^* -d.  
\end{equation}
Overall, the above two inclusions are exactly the KKT conditions of the  VI($K$,$F$) \cite{facchinei2003finite}, hence $\beta^*$ is the v-GNE of the game $G$.
\end{proof}

\begin{proof}[Proof of Lemma~\ref{lemma_tilde_A}]
The mapping $\tilde{\mathcal{A}}$ is $\tilde\epsilon$-cocoercive if and only if
\small
\begin{equation*}
\left\langle \begin{bmatrix}
\beta-\beta'\\
\sigma-\sigma'
\end{bmatrix},\tilde{\mathcal{A}}(\beta,\sigma)-\tilde{\mathcal{A}}(\beta',\sigma')\right\rangle \geq \tilde\epsilon\left\|\tilde{\mathcal{A}}(\beta,\sigma)-\tilde{\mathcal{A}}(\beta',\sigma')\right\|^2
\end{equation*}
\normalsize
for all $\beta,\beta'\in \Omega$, $\sigma,\sigma'\in \mathbb{R}^N$ and some $\tilde \epsilon>0$. By using the definition of $\hat F $ and \eqref{eq_g}, the above inequality is equivalent to
\begin{equation}\label{eq_F_sigma}
  \langle y,Ry\rangle \geq \tilde\epsilon\|Ry\|^2,
\end{equation}
where $y:=\col\left(\begin{bmatrix}
\beta_n-\beta'_n\\
\sigma_n-\sigma'_n
\end{bmatrix}\right)_{n\in\mathcal{N}}$, $R:=\blkdiag(R_n)_{n\in\mathcal{N}}$, $R_n:=\begin{bmatrix}
\mu_n & \ell_n\\
-\kappa & \kappa
\end{bmatrix}$. 

Note that the left hand side of \eqref{eq_F_sigma} is equal to $\langle y,Ry\rangle=\frac{1}{2}\langle y,(R+R^\top)y\rangle$, and $R+R^\top=\blkdiag\left(R_n+R_n^\top\right)_{n\in\mathcal{N}}$,
$$
R_n+R_n^\top=\begin{bmatrix}
2\mu_n & \ell_n-\kappa\\
\ell_n-\kappa & 2\kappa,
\end{bmatrix}
$$
 where $\mu_n=2a_n\frac{N-1}{N}+\frac{1}{\alpha N}$ and $\ell_n=-2a_n\frac{N-1}{N}+\frac{N-2}{\alpha N}$. The above matrix is positive definite if and only if $\kappa>0$ and $4\mu_n\kappa-(\ell_n-\kappa)^2>0$, or equivalently
\begin{equation} \label{eq_k1}
\begin{aligned}
    \kappa \in \big((\sqrt{\mu_n}-\sqrt{\mu_n+\ell_n})^2,(\sqrt{\mu_n}+\sqrt{\mu_n+\ell_n})^2\big).
\end{aligned}
\end{equation}
Note that if $\kappa$ satisfies the condition in \eqref{eq_k}, then \eqref{eq_k1} holds for each $n\in\mathcal{N}$.

For the right hand side of \eqref{eq_F_sigma}, we have $\|Ry\|^2=\langle y,R^{\top}Ry\rangle$, where $R^\top R=\blkdiag\left(R^\top_n R_n\right)_{n\in\mathcal{N}}$ with
$$
R_n^\top R_n=\begin{bmatrix}
\mu^2_n+\kappa^2 & \mu_n\ell_n-\kappa^2\\
\mu_n\ell_n-\kappa^2 & \ell^2_n+\kappa^2
\end{bmatrix}.
$$

Note that $\lambda_{\min}(R_n+R_n^\top)=\bar\epsilon_n$ and $\lambda_{\max}(R_n^\top R_n)=\underline\epsilon_n/2$, and since $R$ is a block diagonal matrix, we have
$$
\langle y,Ry\rangle \geq\frac{1}{2} \min_{n\in\mathcal{N}}\bar\epsilon_n \|y\|^2, \ \|Ry\|^2  \leq \frac{1}{2}\max_{n\in\mathcal{N}}\underline \epsilon_n \|y\|^2.
$$
We, therefore, obtain the following inequality
$$\langle y,R y \rangle \geq \frac{\min_{n\in\mathcal{N}}\bar\epsilon_n}{\max_{n\in\mathcal{N}}\underline \epsilon_n}\|R{y}\|^2.$$
We conclude from the above relation that the mapping $\tilde{\mathcal{A}}$ is $\tilde\epsilon$-cocoercive with 
\[
\tilde\epsilon:=\frac{\min_{n\in\mathcal{N}}\bar\epsilon_n}{\max_{n\in\mathcal{N}}\underline \epsilon_n}.
\]
\end{proof}

\begin{proof}[Proof of Lemma \ref{lemma_positive2}]
We prove this lemma in three steps.

At the first step, we show that the mapping $\mathcal{A}$ is $\epsilon$-cocoercive, where $\epsilon=\min\{\tilde\epsilon,1/\lambda_{\max}(L)\}$ and $\tilde\epsilon$ is given by Lemma \ref{lemma_tilde_A}.
Since the matrix $L_{\lambda}$ is symmetric, $L_{\lambda}\lambda$ is the gradient of $\tilde f(\lambda):=\frac{1}{2}\lambda^\top L_{\lambda} \lambda$, which 
is convex since $\nabla^2 \tilde f(\lambda)=L_{\lambda}$ is positive semidefinite. By Baillon-Haddad theorem (Corollary 18.16 in \cite{bauschke2011convex}), for any $\lambda,\lambda'$, we have
\begin{equation*}
    (L_{\lambda}(\lambda-\lambda'))^\top(\lambda-\lambda')\geq\frac{1}{\lambda_{\max}(L_{\lambda})}\|L_{\lambda}(\lambda-\lambda')\|^2
\end{equation*}
Combining the above inequality and Lemma \ref{lemma_tilde_A}, we have that for any $\omega,\omega'\in\mathbb{R}^{N(3+2M)}$, 
\begin{equation*}
\begin{aligned}
    & \left\langle  \omega-\omega',\mathcal{A}(\omega)-\mathcal{A}(\omega')\right \rangle = (\lambda-\lambda)^\top L_{\lambda}(\lambda-\lambda') \\
     + &(\tilde{\mathcal{A}}(x,\sigma)-\tilde{\mathcal{A}}(x',\sigma'))^\top[x-x',\sigma-\sigma']^\top \\
    \geq  & \tilde \epsilon \|\tilde{\mathcal{A}}(x,\sigma)-\tilde{\mathcal{A}}(x',\sigma')\|^2 + 1/\lambda_{\max}(L)\|L_{\lambda}(\lambda-\lambda')\|^2 \\
    \geq  & \epsilon \|\omega-\omega')\|^2.
\end{aligned}
\end{equation*}
Consequently, we obtain that the mapping $\mathcal{A}$ is $\epsilon$-cocoercieve. 

At the second step, we show that the mapping $\Phi^{-1}\mathcal{A}$ is $\xi$-cocoercive, $\mathcal{U}_{\Phi}$ is $\frac{1}{2\xi}$-averaged in the $\Phi$-induced norm, where $\xi=\frac{\epsilon}{\lambda_{\max}(\Phi^{-1})}$. The mapping $\Phi^{-1}\mathcal{A}$ is $\xi$-cocoercive in the $\Phi$-induced norm if and only if
\begin{equation}
\begin{aligned}
&\langle\Phi^{-1}\mathcal{A}(\omega)-\Phi^{-1}\mathcal{A}(\omega'),\omega-\omega'\rangle_\Phi \\ \geq &\xi\|\Phi^{-1}\mathcal{A}(\omega)-\Phi^{-1}\mathcal{A}(\omega')\|^2_\Phi.
\end{aligned}
\end{equation}
We first provide an upper bound for the right hand side of the above inequality:
\begin{equation}
\begin{aligned}
&\|\Phi^{-1}\mathcal{A}(\omega)-\Phi^{-1}\mathcal{A}(\omega')\|^2_\Phi\\
=&\langle\Phi\Phi^{-1}\mathcal{A}(\omega)-\Phi^{-1}\mathcal{A}(\omega'),\Phi^{-1}\mathcal{A}(\omega)-\Phi^{-1}\mathcal{A}(\omega')\rangle\\
=&(\mathcal{A}(\omega)-\mathcal{A}(\omega'))^\top \Phi^{-1} (\mathcal{A}(\omega)-\mathcal{A}(\omega')) \\
\leq &\|\Phi^{-1}\| \|\mathcal{A}(\omega)-\mathcal{A}(\omega')\|^2\\
=&\lambda_{\max}(\Phi^{-1}) \|\mathcal{A}(\omega)-\mathcal{A}(\omega')\|^2,
\end{aligned}
\end{equation}
where $\lambda_{\max}(\Phi^{-1})$ is largest eigenvalue of $\Phi^{-1}$. Then, we use the $\epsilon$-cocoerciveness of $\mathcal{A}$ and the above derived upper bound to define a cocoercivity constant $\xi$:
\begin{align*}
&\langle\Phi^{-1}\mathcal{A}(\omega)-\Phi^{-1}\mathcal{A}(\omega'),\omega-\omega'\rangle_\Phi\\=&\langle \mathcal{A}(\omega)-\mathcal{A}(\omega'),\omega-\omega'\rangle\\
\geq& \epsilon \|\mathcal{A}(\omega)-\mathcal{A}(\omega')\|^2 \\
\geq&\frac{\epsilon}{\lambda_{max}(\Phi^{-1})}\|\Phi^{-1}\mathcal{A}(\omega)-\Phi^{-1}\mathcal{A}(\omega')\|^2_\Phi.
\end{align*}
Hence, the mapping $\Phi^{-1}\mathcal{A}$ is cocoercive under the $\Phi$-induced norm and $\mathcal{U}_{\Phi}$ is $\frac{1}{2\xi}$-averaged  \cite[Proposition 4.33]{bauschke2011convex}. 

At the last step, we show the iteration $\mathcal{V}_{\Phi}\circ\mathcal{U}_{\Phi}$ is $\theta$-averaged, with $\theta=\frac{1}{2-1/(2\xi)}$, since $\mathcal{U}_{\Phi}$ is $\frac{1}{2\xi}$-averaged and $\mathcal{V}_{\Phi}$ is $\frac{1}{2}$-averaged \cite[Proposition 2.4]{combettes2015compositions}. To ensure $\theta\in(0,1)$, we must have $\xi>\frac{1}{2}$, or equivalently 
$\lambda_{\min}(\Phi)>\frac{1}{2\epsilon}$. The latter holds if and only if the matrix $\Phi-\frac{1}{2\epsilon} I$ is positive definite, i.e., the matrix 
\footnotesize
$$\begin{bmatrix}
\tau^{-1}-\frac{1}{2\epsilon} I & 0 & 0 & 0 & -\bar{A}^\top\\
0 & \upsilon^{-1}-\frac{1}{2\epsilon} I & L_{\sigma} & 0 & 0\\
0 & L_{\sigma} & \rho^{-1}-\frac{1}{2\epsilon} I & 0 & 0\\
0 & 0 & 0 & \delta^{-1}-\frac{1}{2\epsilon} I & L_{\lambda}\\
-\bar{A} & 0 & 0 & L_{\lambda} & \eta^{-1}-\frac{1}{2\epsilon} I 
\end{bmatrix}$$
\normalsize
is positive definite.

Based on the Schur complement argument, the above matrix is positive definite if and only if
\begin{equation}\label{eq_schur_complement1}
\tau^{-1}\succ \frac{1}{2\epsilon} I, \upsilon^{-1} \succ  \frac{1}{2\epsilon} I, \delta^{-1}\succ  \frac{1}{2\epsilon} I,
\end{equation}

\begin{equation}\label{eq_schur_complement2}
(\rho^{-1}-\frac{1}{2\epsilon} I)-L_{\sigma}(\upsilon^{-1}-\frac{1}{2\epsilon} I)^{-1} L_{\sigma}\succ 0,
\end{equation}
\footnotesize
\begin{equation}\label{eq_schur_complement3}
(\eta^{-1}-\frac{1}{2\epsilon}I) -(\bar{A}(\tau^{-1}-\frac{1}{2\epsilon} I)^{-1} \bar{A}^{\top}+L_{\lambda}(\delta^{-1}-\frac{1}{2\epsilon} I)^{-1} L_{\lambda})\succ 0.
\end{equation}
\normalsize
The first relation \eqref{eq_schur_complement1} holds as $\tau_n$, $\upsilon_n$  and $\epsilon_n$ satisfy \eqref{eq_stepsize1} for all $n\in\mathcal{N}$. To verify the second inequality \eqref{eq_schur_complement2}, we note that
\begin{equation}
\begin{aligned}
    L_{\sigma}(\upsilon^{-1}-\frac{1}{2\epsilon} I)^{-1} L_{\sigma}  & \prec (\frac{1}{\max_{n\in\mathcal{N}}\upsilon_n}-\frac{1}{2\epsilon})^{-1}  L^2_{\sigma} \\
    & \prec (\frac{1}{\max_{n\in\mathcal{N}}\upsilon_n}-\frac{1}{2\epsilon})^{-1}  \lambda_{\max}^2(L) I.
\end{aligned}    
\end{equation}
The above together with \eqref{eq_stepsize2} imply that the inequality \eqref{eq_schur_complement2} holds. Finally, to verify \eqref{eq_schur_complement3}, we note that
\begin{equation}
\begin{aligned}
   &\bar{A}(\tau^{-1}-\frac{1}{2\epsilon} I)^{-1} \bar{A}^{\top}+L_{\lambda}(\delta^{-1}-\frac{1}{2\epsilon} I)^{-1} L_{\lambda}   \prec\\
    & (\frac{1}{\max_{n\in\mathcal{N}}\tau_n}-\frac{1}{2\epsilon}) \bar{A} \bar{A}^{\top}+ (\frac{1}{\max_{n\in\mathcal{N}}\delta_n}-\frac{1}{2\epsilon}) L^2_{\lambda}  \prec  \\
    &  (\frac{1}{\max_{n\in\mathcal{N}}\tau_n}-\frac{1}{2\epsilon}) \|\bar A\|^2I+ (\frac{1}{\max_{n\in\mathcal{N}}\delta_n}-\frac{1}{2\epsilon})\lambda_{\max}^2(L) I.
\end{aligned}    
\end{equation}
The above together with \eqref{eq_stepsize3} imply that \eqref{eq_schur_complement3} holds. 
The inequalties \eqref{eq_schur_complement1}, \eqref{eq_schur_complement2}, \eqref{eq_schur_complement3} combined result in $\Phi-\frac{1}{2\epsilon}I$ is positive definite, which completes the proof.
\end{proof}

\end{document}